\documentclass[runningheads]{llncs}

\usepackage[T1]{fontenc}
\usepackage{graphicx}
\usepackage{boxedminipage}
\usepackage[noend]{algpseudocode}
\newcommand{\opt}{\ensuremath{\mathtt{opt}}}
\usepackage{comment}
\usepackage{pgfplots}
\pgfplotsset{compat = newest}
\usepackage{amsmath}
\usepackage{amssymb}
\usepackage{multirow}
\usepackage{threeparttable}

\begin{document}
\title{Budget Feasible Mechanisms for Procurement Auctions with Divisible Agents}

\titlerunning{Budget Feasible Mechanisms with Divisible Agents}

\author{Sophie Klumper\inst{1,2}\orcidID{0000-0002-2375-5313} \and \\ Guido Sch\"{a}fer\inst{1,3} \orcidID{
0000-0002-1923-4902}}
\authorrunning{S.~Klumper and G.~Sch\"{a}fer}

\institute{Networks and Optimization Group, Centrum Wiskunde \& Informatica (CWI)
\and
Department of Mathematics, Vrije Universiteit Amsterdam
\and
Institute for Logic, Language and Computation, University of Amsterdam \\ The Netherlands \\
\email{\{s.j.klumper,g.schaefer\}@cwi.nl}\\
}

\maketitle             

\begin{abstract} 
We consider budget feasible mechanisms for procurement auctions with additive valuation functions. For the divisible case, where agents can be allocated fractionally, there exists an optimal mechanism with approximation guarantee $e/(e-1)$ under the small bidder assumption. We study the divisible case without the small bidder assumption, but assume that the true costs of the agents are bounded by the budget. This setting lends itself to modeling economic situations in which the goods represent time and the agents' true costs are not necessarily small compared to the budget. Non-trivially, we give a mechanism with an approximation guarantee of 2.62, improving the result of 3 for the indivisible case. Additionally, we give a lower bound on the approximation guarantee of 1.25. We then study the problem in more competitive markets and assume that the agents' value over cost efficiencies are bounded by some $\theta \ge 1$. For $\theta \le 2$, we give a mechanism with an approximation guarantee of 2 and a lower bound of 1.18. Both results can be extended to settings with different agent types with a linear capped valuation function for each type. Finally, if each agent type has a concave valuation, we give a mechanism for which the approximation guarantee grows linearly with the number of agent types.

\end{abstract}

\section{Introduction}

We consider procurement auctions in which the auctioneer has a budget that limits the total payments that can be paid to the agents. In this setting, we are given a set of agents $A = [n]$ offering some service (or good), where each agent $i \in A$ has a privately known cost $c_i$ and a publicly known valuation $v_i$. The auctioneer wants to do business with these agents and needs to decide with which agents $S \subseteq A$ to do so. In order for each agent $i \in S$ to comply, the auctioneer will have to make a payment $p_i$ to this agent $i$. The auctioneer has a total budget $B$ available for these payments. The goal of the auctioneer is to find a subset of agents $S \subseteq A$ that (approximately) maximizes the total value, while the total of the payments is at most $B$; such payments are said to be \emph{budget feasible}. As the costs of the agents are assumed to be private, agents may misreport their actual costs to their own advantage. The goal is to design a mechanism that computes this subset and budget feasible payments, such that the agents have no incentive to misreport their actual costs (i.e., the agents are truthful).

Basically, there are two standard approaches in the literature to study this problem: (i) in the \emph{Bayesian setting} it is assumed that the distributions of the agents’ true costs are known, and (ii) in the \emph{prior-free setting} it is assumed that nothing is known about these distributions. In this paper, we focus on the prior-free setting as it may not be possible to extract representative distributions of the agents' true costs. 

Notably, most previous studies focus on the \emph{indivisible case}, where the services (or goods) offered by the agents need to be allocated integrally. In contrast, the \emph{divisible case}, where the agents can also be allocated fractionally, received much less attention and has been studied only under the so-called \emph{small bidder assumption}, i.e., when the agents' costs are much smaller than the available budget. 
While this assumption is justified in certain settings (e.g., for large markets), it is less appropriate in settings where the agents' costs may differ vastly or be close to the budget. 
This motivates the main question that we address in this paper: 
Can we derive budget feasible truthful mechanisms with attractive approximation guarantees if the agents are divisible?

Singer \cite{singer} initiated the study of budget feasible mechanisms in the prior-free setting with indivisible agents and designed a deterministic mechanisms with an approximation guarantee of 5. Later, Chen et al.~\cite{chen} improved this to $2 + \sqrt{2}$. Almost a decade later, Gravin et al.~\cite{gravin} improved the approximation guarantee to $3$ and showed that no mechanism can do better than $3$ (when compared to the fractional optimal solution). 
All three papers also provide additional results such as randomized algorithms with improved approximation guarantees and mechanisms for the more general setting of submodular functions.

To the best of our knowledge, the divisible case of knapsack procurement auctions has not been studied so far. When an agent is offering a service, which can be interpreted as offering time, it is suitable to model this as a divisible good. As an example, consider the situation where an auctioneer has a budget available and wants to organize a local comedy show. Agents are offering to perform and for each agent the auctioneer has a valuation, which reflects the amusement when this agent performs. The agents have costs related to their performance which could consist of the invested time in their performance, attributes needed, travel costs, etc. In this example it makes sense for the auctioneer to have the option to select agents fractionally. After selecting one agent, the budget left might only be enough to let some agent perform half of what they are offering. Or the auctioneer might want to have at least three performances and must select agents fractionally to achieve this due to the budget constraint.

As mentioned, the prior-free setting with divisible goods has been studied by Anari et al.~\cite{Anari} under the small bidder assumption. More formally, if $c_{\max} = \max_{i \in A}\{c_i\}$ and $r = c_{\max}/B$, then the results are analysed for $r \rightarrow 0$. Anari et al.~\cite{Anari} give optimal mechanisms for both the divisible (deterministic) and indivisible (randomized) case, both with an approximation guarantee of $e/(e-1)$. They also mention that in the divisible case, no truthful mechanism with a finite approximation guarantee exists without the small bidder assumption. This can already be shown by an instance with budget $B$ and one agent with value $v$ and true cost $c>B$. An optimal fractional solution can achieve a value of $v\frac{B}{c}$. Therefore any truthful mechanism must allocate the agent fractionally in order to achieve some positive value and thus a finite approximation guarantee. Additionally, it must be that $p \le B$, with $p$ the payment of the agent, for the mechanism to be budget feasible. As we do not know $c$ or any bound on $c$, there is no way to bound $p$ and satisfy budget feasibility in this instance and at the same time achieve the same approximation guarantee in any other possible instance.

However, if we assume that the true costs of the agents are bounded by the budget, there is still a lot we can do without the small bidder assumption and there are settings in which it makes sense to have this assumption. If we revisit our example of an auctioneer wanting to organize a local comedy show, an internationally famous comedian with a true cost exceeding the budget will most likely not be one of the agents offering to perform. On the other hand, there might be a national well-known comedian offering to perform. This comedian might have a true cost that is smaller than the budget, but not much smaller than the available budget, which is what the small bidder assumption requires. 

\paragraph{Our Contributions.}

For the knapsack procurement auction with divisible goods and true costs bounded by the budget, we give a mechanism with an approximation guarantee of 2.62 in Section 3. Additionally, we proof that no mechanism can achieve an approximation guarantee better than 1.25. Although the divisible case gives more freedom in designing an allocation rule, improving the approximation guarantee compared to the indivisible case is non-trivial. In particular, it remains difficult to bound the threshold bids and the complexity of the payment rule increases compared to the indivisible case. Additionally, if an agent is allocated fractionally one needs to determine this fraction exactly while remaining budget feasible. A natural method is to determine this fraction based on the agents' declared costs, but then one must limit the influence that this method gives to the agents.

Proving the above mentioned lower bound of 1.25 on the approximation guarantee requires that the true costs of the agents differ significantly ($\epsilon\ll B$), a property that is used more often when proving lower bounds. 
Therefore, in Section 4, we introduce a setting in which the agents' efficiencies (i.e., value over cost ratios) are bounded by some $\theta \ge 1$. It is reasonable to assume that the efficiencies of the agents is somewhat bounded, as agents will cease to exist if they cannot compete with the other agents in terms of efficiency. Another interpretation of this setting is some middle ground between the prior-free and the Bayesian setting. It might be impossible to find representative distributions of the agents' true costs and address the problem in the Bayesian setting, but the auctioneer could have information about the minimum and maximum efficiencies in the market by previous experiences or market research. We give a mechanism with an approximation guarantee of 2 when the efficiencies are bounded by a factor $\theta \le 2$. For this case we also prove a lower bound of 1.18 on the approximation guarantee and generalize this for different values of $\theta$.

In Section 5, we extend our results to a model in which agents may have different types. This new model allows us to capture slightly more detailed settings. If we revisit the example of hosting a comedy show, this translates to each agent having a certain type of comedy that they preform. In order to set up a nice diverse program, the auctioneer wishes the jokes of a certain type to be limited. The auctioneer knows that, if at some point in time too many jokes of the same type are told, no additional value is added to the show. This is modeled by a linear capped valuation function for each type. We prove that for linear capped valuation functions, the mechanism of Section 3 can be slightly altered to achieve the same approximation guarantee.

Finally, to extend the model of Section 5 with diminishing returns, we assume each agent type has a concave (non-decreasing) valuation function in Section 6. This assumption increases the difficulty of bounding the payments and, as a result, the mechanism we give has an approximation guarantee that grows linearly with the number of agent types. If, however, there is only one type of agent, the mechanism of Section 3 can still be applied to obtain an approximation guarantee of 2.62. 

\begin{table}[t]
\setlength{\tabcolsep}{5pt}
\renewcommand{\arraystretch}{1.1}
\begin{center}
\caption{{\footnotesize Overview of the results obtained in this paper.}}
\begin{threeparttable}[hb]
\begin{tabular}{|l|l|c|c|c|} 
\hline
\multicolumn{2}{|l|}{\multirow{2}{*}{\bf Assumptions}} & \multirow{2}{*}{\bf Origin} & \multicolumn{2}{|c|}{\bf Approximation Guarantee}\\
\cline{4-5}
 \multicolumn{2}{|l|}{} & & Upper bound & Lower bound \\
\hline\hline
\multicolumn{2}{|l|}{\multirow{3}{*}{Indivisible}} & \cite{singer} & 5 & 2 \\
\multicolumn{2}{|l|}{} & \cite{chen} & $\approx 3.41$ & $\approx 2.41$  \\
\multicolumn{2}{|l|}{} & \cite{gravin} & 3 & $3^*$ \\ 
\hline\hline
Divisible & none or (1) & Sections 3 \& 5 & $\approx 2.62$ & $1.25$  \\ 
($c_i \le B\ \forall i$)  & $\theta \le 2$ & Section 4 & 2 & $\approx 1.18$ \\ 
 & (2)  & Section 6 & $O(t)$ & $1.25$  \\ 
\hline
\end{tabular}
\begin{tablenotes}
{\footnotesize
\item \hspace{0.5pt} $^*$Compared to the optimal fractional solution
\item (1) Multiple linear capped valuation functions
\item (2) Multiple concave valuation functions}
\end{tablenotes}
\end{threeparttable}
\end{center}
\vspace*{-.5cm}
\end{table}

\paragraph{Related Work.}

As mentioned earlier, settings with different valuation functions have been studied for the indivisible case. In the case of submodular valuation functions, Singer \cite{singer} gave a randomized mechanism with an approximation guarantee of 112. Again, Chen et al.~\cite{chen} improved this to $7.91$ and gave a deterministic exponential time mechanism with an approximation guarantee of $8.34$. Later Jalaly and Tardos~\cite{Jalaly} improved this to 5 and 4.56 respectively. For subadditive valuation functions, Dobzinski et al.~\cite{Dobzinski} gave a randomized and deterministic mechanism with an approximation guarantee of $O(\log^{2}n)$ and $O(\log^{3}n)$ respectively. Bei et al.~\cite{Bei} improved this to $O(\log n /( \log \log n))$ with a polynomial time randomized mechanism and gave a randomized exponential time mechanism with an approximation guarantee of 768 for XOS valuation functions. Leonardi et al.~\cite{Leonardi} improved the latter to 436 by tuning the parameters of the mechanism. 

Related settings have also been studied for the indivisible case. Leonardi et al.~\cite{Leonardi} consider the problem with an underlying matroid structure, where each element corresponds to an agent and the auctioneer can only allocate an independent set. Chan and Chen~\cite{Chan} studied the setting in which agents offer multiple units of their good. They regard concave additive and subadditive valuation functions. The setting in which the auctioneer wants to get a set of heterogeneous tasks done and where each task requires the performing agent to have a certain skill has been studied by Goel et al.~\cite{Goel}. They give a randomized mechanism with an approximation guarantee of 2.58, which is truthful under the small bidder assumption. Jalaly and Tardos~\cite{Jalaly} match this result with a deterministic mechanism. The results of Goel et al.~\cite{Goel} can be extended to settings in which tasks can be done multiple times and agents can perform multiple tasks. Related to this line of work is also the strategic version of matching and coverage, in which edges and subsets represent strategic agents, that Singer~\cite{singer} also studied. Chen et al.~\cite{chen} also studied the knapsack problem with heterogeneous items, were items are divided in groups and at most one item from each group can be allocated. Amanatidis et al.~\cite{Amanatidis} give randomized and deterministic mechanisms for a subclass of XOS problems. For the mechanism design version of the budgeted max weighted matching problem, they give a randomized (deterministic) mechanism with an approximation guarantee of 3 (4) and they generalize their results to problems with a similar combinatorial structure.

\section{Preliminaries}

We are given a (finite) set of agents $A$ offering some service (or good), where each agent $i \in A$ has a privately known non-negative cost $c_i$ and a non-negative valuation $v_i$. 
Each agent $i \in A$ declares a non-negative cost $b_i$, which they might use to misreport their actual cost $c_i$. 
Given the declared costs $\mathbf{b} = (b_i)_{i \in A}$, valuations $\mathbf{v} = (v_i)_{i \in A}$ and a budget $B > 0$, the problem is to design a mechanism that computes an allocation vector $\mathbf{x} = (x_i)_{i \in A}$ and a payment vector $\mathbf{p} = (p_i)_{i \in A}$.\footnote{It is important to realize that the mechanism only has access to the declared costs $\mathbf{b}$, as the actual costs $\mathbf{c}$ are assumed to be private information of the agents.}
The allocation vector $\mathbf{x}$ should satisfy $x_i \in [0,1]$ for all $i \in A$, where $x_i$ denotes the fraction selected of agent $i$. An element $p_i$ of the payment vector $\mathbf{p}$ corresponds to the payment of agent $i$. We want to (approximately) maximize the value of the allocation vector and the total of the payments to be within the budget. We assume that the true costs $\mathbf{c} = (c_i)_{i \in A}$ are bounded by the budget, i.e., $c_{i} \leq B$ for all $i \in A$. We also assume that the costs incurred by the agents are linear, i.e., given allocation vector $\mathbf{x}$ the cost incurred by agent $i$ is equal to $c_{i} x_{i}$. Therefore the utility $u_i$ of agent $i \in A$ is equal to $u_{i} = p_{i} - c_{i}  x_{i}$. The auctioneer has an additive valuation function, i.e., given allocation vector $\mathbf{x}$ the value derived by the auctioneer is $v(\mathbf{x}) = \sum_{i \in A} v_{i}  x_{i}$. The goal of each player $i \in A$ is to maximize their utility $u_i$, and as the costs of the agents are assumed to be private, they may misreport their actual costs to achieve this. If the agents are not strategic, i.e., the costs are publicly known, the above setting naturally corresponds to the fractional knapsack problem. 

We seek mechanisms that satisfy the following properties:
\begin{enumerate}
    \item \emph{Truthfulness:} for every agent $i$ reporting their true cost is a dominant strategy: for any declared costs $b_i$ and $\mathbf{b}_{-i}$ it holds that $p_i - x_i c_i \ge p'_i - x'_i c_i$, where ($x_i$, $p_i$) and ($x'_i$, $p'_i$) are the allocations and payments of agent $i$ with respect to the declared costs $(c_i, \mathbf{b}_{-i})$ and $(b_i, \mathbf{b}_{-i})$.
    \item \emph{Individual rationality:} for every agent $i$ it holds that $p_i \ge x_i b_i$, so under a truthful report agent $i$ has non-negative utility.
    \item \emph{Budget feasibility:} the total payment is at most the budget: $\sum_{i \in A} p_i \le B$.
    \item \emph{Approximation guarantee:} a mechanism has an approximation guarantee of $\gamma \ge 1$ if, for any declared costs $\mathbf{b}$, it outputs an allocation $\mathbf{x}$ such that $\gamma \cdot v(\mathbf{x}) \ge \opt$, where $\opt$ is the value of the optimal fractional solution with respect to $\mathbf{b}$.
    \item \emph{Computational efficiency:} the allocation and payment vector can be computed in polynomial time.
\end{enumerate}

Given declared costs $\mathbf{b} = (b_i)_{i \in A}$, let $x_i(\mathbf{b})$ be the fraction selected of agent $i \in A$. An allocation rule is said to be \emph{monotone non-increasing} if for each agent $i$, the fraction of $i$ selected can only increase as their declared cost decreases. More formally, for all $\mathbf{b}_{-i}$ and $b_i > b'_i$: $x_i(b_i, \mathbf{b}_{-i}) \le x_i(b'_i, \mathbf{b}_{-i})$. In order to design truthful mechanisms, we will exploit Theorem \ref{Tardos:Payments} and derive mechanisms that have a monotone non-increasing allocation rule. 

\begin{theorem}\label{Tardos:Payments}
(Archer and Tardos~\cite{tardos}) A monotone non-increasing allocation rule $x(\mathbf{b})$ admits a truthful payment rule that is individually rational if and only if for all $i, \mathbf{b}_{-i}$: $\int_{0}^{\infty} x_{i}(u,\mathbf{b}_{-i}) du < \infty$. In this case we can take the payment rule $p(\mathbf{b})$ to be
\begin{equation}
p_{i}(b_{i},\mathbf{b}_{-i}) = b_{i} \cdot x_{i}(b_{i},\mathbf{b}_{-i}) + \int_{b_{i}}^{\infty} x_{i}(u,\mathbf{b}_{-i}) du \quad \forall i.
\label{payments}
\end{equation}
\end{theorem}

Next we introduce some additional notation that is used throughout the paper. We say that agent $i$ \emph{wins} if $x_i >0$ and \emph{loses} if $x_i = 0$. Note that $p_i =0$ for a losing agent $i$ if the above payment rule is used. We define the \textit{efficiency} of an agent as the ratio of their value over (declared) cost.
Agents with a high efficiency are preferable as, compared to agents with a lower efficiency, they (relatively) contribute a higher value per unit cost. If the cost of an agent $i$ is $0$, we define their efficiency as being $\infty$ such that $i$ has maximum efficiency among the considered agents. Whenever we order the agents according to decreasing efficiencies, we assume that ties are broken arbitrarily but consistently. We define $\opt(A,\mathbf{c})$ as the value of the optimal fractional solution regarding the set of agents $A$, costs $\mathbf{c} = (c_i)_{i \in A}$ and values $\mathbf{v} = (v_i)_{i \in A}$. $\opt_{-i}(A,\mathbf{c})$ is defined similarly, only regarding the set of agents $A \setminus \{i\}$. For the sake of readability, we omit the valuations $\mathbf{v}$ as an argument as these are publicly known. 

Given that we focus on the design of truthful mechanisms in this paper, we adopt the convention (which is standard in this context) and often use $\mathbf{c} = (c_i)_{i \in A}$ also to refer to the declared costs of the agents.

\section{Linear Valuation Functions}

\paragraph{The Mechanism.}

It is well-known that designing a truthful and individually rational mechanism comes down to designing a monotone allocation rule and implementing it together with the payment rule of Theorem \ref{Tardos:Payments}. When selecting an agent $i$ in this case, the first term of (\ref{payments}) is equal to cost incurred by agent $i$ and the second term can be interpreted as the amount agent $i$ is overpaid. It is beneficial for this second term to be as small as possible, in order to select more agents and achieve a better approximation guarantee. 

We give a mechanism that realizes this by imposing a \textit{threat}, denoted $\tau_i$, on each agent $i$ in the second part of the mechanism. 
If the declared cost of agent $i$ exceeds this threat, $i$ will loose. In this case, the second term of the payment formula is equivalent to $\int_{b_{i}}^{\tau_{i}} x_{i}(u,\mathbf{b}_{-i}) du$, and our goal thus is to choose $\tau_i$ close to the true cost $c_i$ to not overpay too much, while remaining at least as large to avoid a negative impact on the approximation guarantee.

In order to use these threats to bound the payments, $\tau_i$ must not increase when agent $i$ declares a higher cost. Note that bounding the payments with $\tau_i$ can still be done if $\tau_i$ decreases when agent $i$ declares a higher cost. In our mechanism, we use the following threat $\tau_i$ for an agent $i$, which is independent of the declared cost of $i$:
\[ \tau_i = v_i \frac{B}{\alpha(1+\beta) \opt_{-i}(N, \mathbf{c})}. \]
Note that this threat $\tau_i$ imposes an upper bound on the \textit{threshold bid} of agent $i$, i.e., the largest cost agent $i$ can declare such that $i$ wins. 

Additionally, we define for each agent $i$:
\[ \rho_{i}(N,\mathbf{c}) = \frac{v_i}{\opt_{-i}(N, \mathbf{c})},\]
which represents some measure of how valuable agent $i$ individually is. Note that an agent $i$ has no influence on their own ratio. 

Our allocation rule will either select one valuable agent, or select agents in a greedy manner according to efficiencies, which naturally leads to a monotone allocation rule. Let $\alpha \in (0,1]$ and $\beta > 0$ be some parameters which we fix later. Our mechanism is as follows: 

\begin{center}
\begin{boxedminipage}{11cm}
\textsc{\textbf{Divisible Agents (DA)}}
\begin{algorithmic}[1]
\small
\State Let $N=\{i \in A \; |\; c_i \le B\}$, $n = |N|$ and $i^* = \arg\max_{i\in N} \rho_{i}(N,\mathbf{c}) $
\If{$ \rho_{i^*}(N,\mathbf{c}) \ge \beta$} set $x_{i^*}=1$ and $x_{i}=0$ for $i\in N \setminus \{i^*\}$
\Else
\State Rename agents s.t. $\frac{v_1}{c_1} \ge \frac{v_2}{c_2} \ge \dots \ge \frac{v_n}{c_n}$
\State Compute $\mathbf{x}$ s.t. $v(\mathbf{x}) = \alpha \opt(N, \mathbf{c}) $ with $x_{i} = 1$ for $i<k$, $x_{k} \in (0,1]$

and $x_{i} = 0$ for $i>k$, $k \le n $
\For{$i \le k$} 
\If{$c_i > \tau_i $} set $x_i = 0$ \EndIf 
\EndFor
\EndIf
\State For $i \in N$ compute payments $p_i$ according to (\ref{payments})
\State \textbf{return} $(\mathbf{x}, \mathbf{p})$
\end{algorithmic}
\end{boxedminipage}
\end{center}

\begin{lemma}\label{TrInRa}
Mechanism \textsc{\textbf{DA}} is truthful and individually rational. 
\end{lemma}

\begin{proof}
We start by showing that the allocation rule is monotone non-increasing. Suppose that, given declared costs $\mathbf{c} = (c_i)_{i \in A}$, $\mathbf{x}$ is computed by the mechanism. 

Suppose the mechanism selected agent $i^*$ and suppose $i^*$ decreases their declared cost to $c' < c_{i^*}$ and let $\mathbf{c'} = (c', \mathbf{c}_{-i^*})$. Note that the set $N$ does not change as long as $c' \le B$. The values $\rho_{i}(N, \mathbf{c'})$ of agents $i \neq i^*$ are the same or decrease, the value $\rho_{i^*}(N, \mathbf{c'})$ of $i^*$ is unaffected and therefore $i^*$ remains fully selected. Now suppose $i^*$ increases their declared cost $c'$ such that $c_{i^*} < c' \le B$. The values $\rho_{i}(N, \mathbf{c'})$ of agents $i \neq i^*$ are the same or increase and again the value $\rho_{i^*}(N, \mathbf{c'})$ of $i^*$ is unaffected. Therefore $i^*$ remains fully selected if $\forall i \in N$: $\rho_{i^*}(N, \mathbf{c'}) \geq \rho_{i}(N, \mathbf{c'})$, otherwise $i^*$ loses. When $i^*$ increases their declared cost to $c' > B$, $i^*$ definitely loses. 

Otherwise the mechanism computed $\mathbf{x}$ such that (initially) $v(\mathbf{x}) = \alpha \opt(N, \mathbf{c})$. Suppose some winning agent $i$ decreases their declared cost to $c' < c_i$ and let $\mathbf{c'} = (c', \mathbf{c}_{-i})$. The values $\rho_{j}(N,\mathbf{c'})$ of agents $j \neq i$ are the same or decrease and the value $\rho_{i}(N,\mathbf{c'})$ of $i$ is unaffected, so the allocation vector is still computed in the second part of the mechanism. The ratio $\frac{v_i}{c'}$ increases, so $i$ can only move further to the front of the ordering. In addition, $\opt(N, \mathbf{c'})$ can only increase and therefore $i$ will be selected to the same extent or more. Note that the threat will not deselect agent $i$, as the value of the threat does not change and $i$ decreases their declared cost. 
Now suppose $i$ increases their declared cost $c'$ such that $c_i < c' \le B$. The values $\rho_{j}(N,\mathbf{c'})$ of agents $j \neq i$ are the same or increase and the value $\rho_{i}(N,\mathbf{c'})$ of $i$ is unaffected. If the mechanism now selects agent $i^*$, agent $i$ loses. Otherwise, the ratio $\frac{v_i}{c'}$ decreases, so $i$ can only move further to the back of the ordering. In addition, $\opt(N, \mathbf{c'})$ can only decrease and therefore $i$ will be selected to the same extent or less. When $i$ increases their declared cost to $c' > \min\{ B, \tau_{i}\}$, $i$ definitely loses. Therefore, the allocation rule is monotone. 

Note that a winning agent $i$ definitely loses when declaring a cost $c' > B$ and as $x_i \in [0, 1]$, we have $ \int_{0}^{\infty} x_{i}(u,\mathbf{c}_{-i}) du \le 1 \cdot B < \infty$. Therefore, the mechanism is truthful and individually rational by Theorem \ref{Tardos:Payments}. \hfill \qed 
\end{proof}

\begin{lemma}\label{BudFea}
Mechanism \textsc{\textbf{DA}} is budget feasible.
\end{lemma}

\begin{proof}
Suppose, given declared costs $\mathbf{c} = (c_i)_{i \in A}$, $\mathbf{x}$ is computed by selecting agent $i^*$. We have $\sum_{i \in N} p_i = p_{i^*} \le 1 \cdot B$, as $x_{i^*} \in [0, 1]$ and $i^*$ definitely loses when declaring a cost $c' > B$.

Otherwise, the mechanism computed $\mathbf{x}$ by selecting $k$ agents such that $v(\mathbf{x}) \le \alpha \opt(N, \mathbf{c})$. By construction we know that the threshold bid for agent $i$ is smaller than or equal to our threat $\tau_i$. Therefore, and because the allocation rule is monotone, the payment of agent $i$ can be bounded by $x_i \tau_i$ and
\[\sum_{i \in N} p_i = \sum_{i=1}^{k} p_{i} \le \sum_{i=1}^{k} x_i v_i \frac{B}{\alpha(1+\beta) \opt_{-i}(N, \mathbf{c})} \le \sum_{i=1}^{k} x_i v_i \frac{B}{\alpha \opt(N, \mathbf{c})} \le B,\]
were the second inequality follows from $\opt(N, \mathbf{c}) \le \opt_{-i}(N, \mathbf{c}) + v_i  < (1 + \beta)\opt_{-i}(N, \mathbf{c})$, as the mechanism did not select agent $i^*$. Hence, the mechanism is budget feasible. \hfill \qed 
\end{proof}

In order to prove the approximation guarantee, we need the following lemma. 

\begin{lemma}\label{LemLowBouB}
Let $\mathbf{x^*}$ be a solution of $\opt(A, \mathbf{c})$ and assume agents $i\in A$, $n = |A|$, are ordered such that $\frac{v_1}{c_1} \ge \dots \ge \frac{v_n}{c_n}$. Let $k\le n$ be an integer such that $\sum_{i=1}^{k-1} v_i x_{i}^* < \alpha \opt(A, \mathbf{c}) \le \sum_{i=1}^{k} v_i x_{i}^*$, $\alpha \in (0,1]$. Then $\frac{c_{k}}{v_{k}} (1 - \alpha) \opt(A, \mathbf{c}) \le B$.
\end{lemma}

\begin{proof}
For convenience let $\opt = \opt(A, \mathbf{c})$. Note that we can split $\mathbf{x}^*$ into $\mathbf{x} = (x^{*}_{1}, \dots, x^{*}_{k-1}, x_{k}, 0, \dots, 0)$ and $\mathbf{y} = (0, \dots, 0, x^{*}_{k} - x_{k}, x^{*}_{k+1}, \dots, x^{*}_{n})$ with $x_{k} \le x_{k}^*$ such that $\mathbf{x}^* = \mathbf{x} + \mathbf{y}$, $v(\mathbf{x}) = \alpha \opt$ and $v(\mathbf{y}) = (1 - \alpha) \opt$. By feasibility of $\mathbf{x}^*$, and thus $\mathbf{y}$, and the ordering of the agents it follows that
\[ B \ge \sum_{i=k}^{n} c_i y_{i} = \sum_{i=k}^{n} \frac{c_i}{v_i} v_i y_{i} \ge \frac{c_{k}}{v_{k}} \sum_{i=k}^{n} v_i y_{i} = \frac{c_{k}}{v_{k}} (1 - \alpha) \opt.
\] 

\vspace*{-4ex}
\hfill \qed 
\end{proof}

Lemma \ref{LemLowBouB} can be interpreted in the following way. If agents $k$ up to some $j \in A$, $j\ge k$, together contribute a fraction of $(1-\alpha)$ of the value of the optimal solution, then `paying' these agents a cost per unit value of $\frac{c_k}{v_k}$ is budget feasible, as these agents actually have a greater or an equal cost per unit value. 

\begin{lemma}\label{AppGua}
Mechanism \textsc{\textbf{DA}} has an approximation guarantee of $\frac{\sqrt{5}+1}{\sqrt{5}-1}$ if $\alpha = \frac{\sqrt{5} -1}{\sqrt{5} +1}$ and $\beta = \frac{1}{2}\sqrt{5} - \frac{1}{2}$. 
\end{lemma}

\begin{proof}
Suppose given declared costs $\mathbf{c} = (c_i)_{i \in A}$, $\mathbf{x}$ is computed by the mechanism. If the mechanism selected agent $i^*$, we have
\[ v_{i^*} \ge \beta \opt_{-i^*}(N, \mathbf{c}) \ge \beta \opt(N, \mathbf{c}) - \beta v_{i^*} \quad \Leftrightarrow \quad v(\mathbf{x}) \ge \frac{\beta}{1 + \beta} \opt(N, \mathbf{c}).\]

Otherwise the mechanism computed $\mathbf{x}$ by selecting $k$ agents such that initially $v(\mathbf{x}) = \alpha \opt(N, \mathbf{c})$. We want no agent $i \le k$ to be deselected by our threat, i.e., we want $c_i \le \tau_i$. If  
\[ R:= \frac{\alpha(1+\beta)}{(1 - \alpha)} \le 1 \quad \text{ then } \quad \frac{\alpha(1+\beta)}{(1 - \alpha)}  \opt_{-i}(N, \mathbf{c}) \le \opt(N, \mathbf{c}),\]
and rearranging terms leads to 
\[ \tau_i = v_i \frac{B}{\alpha(1+\beta) \opt_{-i}(N, \mathbf{c})} \ge  v_i \frac{B}{(1 - \alpha) \opt(N, \mathbf{c})} \ge c_i, \]
where the last inequality follows from Lemma \ref{LemLowBouB} and the ordering of the agents. To see that the condition of Lemma \ref{LemLowBouB} is satisfied, note that there exists an optimal solution of the form $(1, \dots, 1, a \in [0,1], 0, \dots, 0)$ and for equal values of $\alpha$ the integer $k$ of the mechanism and Lemma \ref{LemLowBouB} then correspond.  

Balancing the approximation guarantees $(1 + \beta)/\beta = 1/\alpha$ subject to $R \le 1$ leads to $\alpha = \frac{\sqrt{5} -1}{\sqrt{5} +1} \approx 0.38$, $\beta = \frac{1}{2}\sqrt{5} - \frac{1}{2}  \approx 0.61$ and an approximation guarantee of $\frac{\sqrt{5}+1}{\sqrt{5}-1} \approx 2.62$. \hfill \qed 
\end{proof}

The mechanism \textsc{\textbf{DA}} is computationally efficient: it is trivial to see that all steps take polynomial time, except for the computation of the payments. The latter can also be done efficiently for each agent $i$ because the payment function is piecewise and the number of subfunctions one needs to consider is bounded by $n$. Further details can be found in Appendix \ref{AppComEff}. From Lemmas \ref{TrInRa}, \ref{BudFea} and \ref{AppGua}, we arrive at the following theorem.

\begin{theorem}\label{Th:DA}
Mechanism \textsc{\textbf{DA}} with $\alpha = \frac{\sqrt{5} -1}{\sqrt{5} +1}$ and $\beta = \frac{1}{2}\sqrt{5} - \frac{1}{2}$ is truthful, individually rational, budget feasible, computationally efficient and has an approximation guarantee of $\frac{\sqrt{5}+1}{\sqrt{5}-1}$.
\end{theorem}

\paragraph{Lower Bound.}

Next we show that no truthful, individually rational, budget feasible and deterministic mechanism exists with an approximation guarantee of $(\frac{5}{4} - \epsilon_1)$, for some $\epsilon_1>0$. For contradiction, assume such a mechanism does exist. Consider two instances, both with budget $B$ and two agents with equal valuations: $\mathcal{I}_1 = (B, \mathbf{v} = \mathbf{1}, \mathbf{c} = (B,B))$ and $\mathcal{I}_2 = (B, \mathbf{v} = \mathbf{1}, \mathbf{c} = (\epsilon_2,B))$. In the first instance $\opt = 1$, so for the approximation guarantee to hold, an allocation vector must satisfy $v(\mathbf{x}) \ge 4/(5-4\epsilon_1)$. Therefore, any such mechanism must have some agent $i$ for which $x_i \ge 2/(5-4\epsilon_1)$ and assume w.l.o.g. that this is agent 1. In the second instance $\opt = 2 - \epsilon_{2}/B$, so for the approximation guarantee to hold, an allocation vector must satisfy $v(\mathbf{x}) \ge (2 - \epsilon_{2}/B)( 4/(5-4\epsilon_1))$. By the previous instance and individual rationality, it follows that agent 1 can guarantee itself a utility of at least $u = (B-\epsilon_{2}) 2/(5-4\epsilon_1)$ by deviating to $B$. As agent 1 must be somewhat selected to achieve the approximation guarantee, $p_1 \ge \epsilon_2 x_1 + u$ by truthfulness. Therefore in the best case, if agent 1 is entirely selected, this leads to a budget left of at most $B' = B - \epsilon_2 - u$. By spending this all on agent 2, this leads to an allocation vector with a value of at most $2 - \epsilon_{2}/B - u/B$. With elementary calculations, one can show that this value is smaller than $\opt / (\frac{5}{4} - \epsilon_1)$ if $\epsilon_{2} < (8B\epsilon_{1})/(4\epsilon_{1} +1)$, resulting in a contradiction.

\section{Competitive Markets}

It is common for lower bound proofs to use multiple instances in order to show that a mechanism cannot satisfy all properties in each instance. In our proof, and also in the lower bound proofs by Singer~\cite{singer} and Gravin et al.~\cite{gravin}, this leads to very specific instances. One instance has agents with equal efficiency while the other instance has agents for which the difference in efficiency goes to infinity. Both instances are plausible in, say, a mature market were the efficiencies of the agents are close, or a premature market were the efficiencies of the agents differ a lot. 

However, it is reasonable to assume that after some time the efficiencies of the agents are somewhat bounded, as agents will cease to exist if they cannot compete with the other agents in terms of efficiency. We therefore introduce a setting in which some bound on the efficiencies of the agents is known and seek a tighter approximation guarantee for this setting. We formalize this with the following definition.

\begin{definition}\label{defTheta}
A procurement auction instance $\mathcal{I} = (B, (v_i)_{i\in A},(c_i)_{i\in A})$ is $\theta$-competitive with $\theta \geq 1$ if
\begin{equation}
\max_{i \in A: c_i \le B} \frac{v_i}{c_i} \le \theta \min_{i \in A: c_i \le B} \frac{v_i}{c_i}.
\label{eqTheta}
\end{equation}
\end{definition}
In this setting, we will also say that the agents are \emph{$\theta$-competitive}. Note that if $\theta = 1$ then all agents are equally competitive. If $\theta \rightarrow \infty$ the competitiveness of the agents is unbounded, which corresponds to the original setting. 

\paragraph{The Mechanism.}

Under the assumption that an instance is $\theta$-competitive, an agent can only increase their declared cost up to some $c'$ before becoming the agent with worst efficiency. Again, we use the payment rule of Theorem \ref{Tardos:Payments}, and want the second term of (\ref{payments}) to be as small as possible. We give a mechanism that realizes this, not by directly imposing a threat, but by setting the parameters of the mechanism to specific values. This will ensure that if, in the second part of the mechanism, the declared cost of an agent exceeds some $c'$, this agent will lose.
Let $\alpha \in (0,1]$ and $\beta > 0$ be some parameters which we fix later. Our mechanism is as follows:

\begin{center}
\begin{boxedminipage}{11cm}
\textsc{\textbf{Divisible $\theta$-competitive Agents (DA-$\theta$)}}
\begin{algorithmic}[1]
\small
\State Let $N=\{i \in A \;|\; c_i \le B\}$, $n = |N|$ and $i^* = \arg\max_{i\in N} \rho_{i}(N,\mathbf{c}) $
\If{$\rho_{i^*}(N,\mathbf{c}) \ge \beta$} set $x_{i^*}=1$ and $x_{i}=0$ for $i \in N\setminus \{i^*\}$
\Else
\State Rename agents such that $\frac{v_1}{c_1} \ge \frac{v_2}{c_2} \ge \dots \ge \frac{v_n}{c_n}$
\State Compute $\mathbf{x}$ s.t. $v(\mathbf{x}) = \alpha \opt(N, \mathbf{c}) $ with $x_{i} = 1$ for $i<k$, $x_{k} \in (0,1]$ 

and $x_{i} = 0$ for $i>k$, $k \le n $
\EndIf
\State For $i \in N$ compute payments $p_i$ according to (\ref{payments})
\State \textbf{return} $(\mathbf{x}, \mathbf{p})$
\end{algorithmic}
\end{boxedminipage}
\end{center}

\begin{lemma}\label{TrInRa2}
Mechanism \textsc{\textbf{DA-$\theta$}} is truthful and individually rational. 
\end{lemma}

The proof of Lemma \ref{TrInRa2} is identical to the proof of Lemma \ref{TrInRa} with one exception: In the second part of the mechanism, a winning agent $i$ will definitely loose if $i$ would have declared a cost $c' > \min\{ B, \theta c_i \}$, were the reasoning of the second argument is stated in the proof of Lemma \ref{BudFea2}. In order to prove budget feasibility, we will need the following lemma which states that if some $k$ most efficient agents together contribute a fraction of $\alpha$ of the value of the optimal solution, then the corresponding total cost of these agents cannot exceed a fraction of $\alpha$ of the budget.

\begin{lemma}\label{boundOptCost}
Let $\mathbf{x^*}$ be a solution of $\opt(A, \mathbf{c})$ and assume agents $i \in A$, $n = |A|$, are ordered such that $\frac{v_1}{c_1} \ge \dots \ge \frac{v_n}{c_n}$. Let $v(\mathbf{x}) = \alpha \opt(A, \mathbf{c})$ with $\alpha \in (0,1]$, integer $k \le n$, $x_i = x^{*}_{i}$ for $i< k$, $x_{k} \in (0,x^*_k]$, $x_i = 0$ for $i>k$. Then $ \sum_{i=1}^{k} c_i x_i \le \alpha B$.
\end{lemma}

\begin{proof}
For convenience let $\opt = \opt(A, \mathbf{c})$. Define $\mathbf{y} = \mathbf{x}^* - \mathbf{x}$, then $v(\mathbf{y}) = (1 - \alpha) \opt$. For contradiction, suppose $c(\mathbf{x}) = \sum_{i=1}^{k} c_i x_i > \alpha B$. Then it must be that $c(\mathbf{y}) < (1-\alpha) B$, as otherwise $c(\mathbf{x}^*) = c(\mathbf{x}) + c(\mathbf{y}) > B$. We have 
\[ \alpha B < \sum_{i=1}^{k} c_i x_i = \sum_{i=1}^{k} \frac{c_i}{v_i} v_i x_i \le \frac{c_{k}}{v_{k}} \sum_{i=1}^{k}  v_i x_i = \frac{c_{k}}{v_{k}} \alpha \opt \quad \Leftrightarrow \quad B < \frac{c_{k}}{v_{k}} \opt, \]
and
\[ (1- \alpha) B > \sum_{i=k}^{n} c_i y_i = \sum_{i=k}^{n} \frac{c_i}{v_i} v_i y_i \ge \frac{c_{k}}{v_{k}} \sum_{i=k}^{n} v_i y_i = \frac{c_{k}}{v_{k}} (1 - \alpha) \opt \quad \Leftrightarrow \quad B > \frac{c_{k}}{v_{k}} \opt, \]
resulting in a contradiction. \hfill \qed 
\end{proof}

\begin{lemma}\label{BudFea2}
Mechanism \textsc{\textbf{DA-$\theta$}} is budget feasible if $\alpha \le \min \big \{ \frac{1}{\theta}, \frac{1}{1+ \beta} \big \}$.
\end{lemma}

\begin{proof}
If, given declared costs $\mathbf{c} = (c_i)_{i \in A}$, $\mathbf{x}$ is computed by selecting agent $i^*$, the proof is identical to Lemma \ref{BudFea}. Otherwise the mechanism computed $\mathbf{x}$ such that $v(\mathbf{x}) = \alpha \opt(N, \mathbf{c})$. Suppose a winning agent $i$ increases their declared cost to $c' > \theta c_i$ and let $\mathbf{c'} = (c', \mathbf{c}_{-i})$. If $c'> B$, agent $i$ loses. If this is not the case, we proof that $i$ loses if $\alpha \le \frac{1}{1+ \beta}$. As the allocation vector was initially computed in the second part of the mechanism, we have $v_i \le \beta \opt_{-i}(N, \mathbf{c})$ leading to
\[ \opt(N, \mathbf{c'}) \le \opt(N, \mathbf{c}) \le \opt_{-i}(N, \mathbf{c}) + v_i  \le (1 + \beta) \opt_{-i}(N, \mathbf{c}) \quad \Leftrightarrow\]
\[ \alpha \opt(N, \mathbf{c'}) \le \frac{1}{1+ \beta} \opt(N, \mathbf{c'}) \le \opt_{-i}(N, \mathbf{c}) \le \sum_{j \in N, j \ne i} v_j.\]
The above inequality points out that the total valuation of agents $j \ne i$ is greater than or equal to $\alpha \opt(N, \mathbf{c'})$. Therefore, as agent $i$ will be last in the ordering when increasing their cost to $c'$, $i$ loses if the allocation vector is computed in the second part of the mechanism. If the allocation vector is computed by selecting $i^*$, agent $i$ also loses. Therefore, and by monotonicity of the allocation rule, we can bound the sum of the payments with
\[\sum_{i \in N} p_i = \sum_{i=1}^{k} p_{i} \le \sum_{i=1}^{k} \theta c_i x_i \le \theta \alpha B,\]
where the last inequality follows from Lemma \ref{boundOptCost}. To see this, note that an optimal solution of the form $(1, \dots, 1, a \in [0,1], 0, \dots, 0)$ exists and for equal values of $\alpha$ the integer $k$ and allocation vector $\mathbf{x}$ of the mechanism and Lemma \ref{boundOptCost} then correspond. So if $\alpha \le \frac{1}{\theta}$ and $\alpha \le \frac{1}{1+ \beta}$, the payments are budget feasible. \hfill \qed 
\end{proof}

The mechanism \textsc{\textbf{DA-$\theta$}} is computationally efficient and details can be found in Appendix \ref{AppComEff}. The approximation guarantee is equal to $ \max \big \{\frac{\beta +1 }{\beta},  \frac{1}{\alpha} \big\}$, where the second argument is by construction and the first argument follows from exactly the same reasoning as in Lemma \ref{AppGua}. Minimizing this max-expression subject to $\alpha \le \min \big \{ \frac{1}{\theta}, \frac{1}{1+ \beta} \big \}$ (budget feasibility) leads to $\alpha = \frac{1}{2}$ and $\beta = 1$ if $\theta \in [1, 2]$. Therefore, the following theorem follows from Lemmas \ref{TrInRa2} and \ref{BudFea2}. 

\begin{theorem} 
If $\theta \in [1, 2]$, mechanism \textsc{\textbf{DA-$\theta$}} with $\alpha = \frac{1}{2}$ and $\beta = 1$ is truthful, individually rational, budget feasible, computationally efficient and has an approximation guarantee of 2.
\end{theorem}

It can be seen in the left plot of Figure \ref{fig:plots} that if $\theta > 2$, it is optimal to set $\alpha = \frac{1}{\theta}$ and $\beta$ to any number in $\big [\frac{\alpha}{1-\alpha}, \frac{1 - \alpha}{\alpha} \big ]$, as then $\alpha$ is the limiting factor in the approximation guarantee ($\frac{1}{\alpha} \ge (1 + \beta)/\beta$) and the value of $\alpha$ is feasible ($\alpha \le 1/(1 +\beta)$). So if it is known that agents are $\theta$-competitive with $\theta < 2.62$, it is beneficial to use mechanism \textsc{\textbf{DA-$\theta$}} instead of mechanism \textsc{\textbf{DA}}. 

\begin{figure}
\hspace*{-.1cm}
\tikzset{every path/.style={line width=1.pt}}
\begin{tikzpicture}[scale=0.7]
\begin{axis}[ axis lines = left, legend style={font=\small}, xlabel = \(\beta\), ylabel style ={rotate=-90}, ylabel = {\(\alpha\)}, legend style = {yshift=0.2cm} ]
\addplot [domain=0:2.4, samples=100, color=black] {x/(1+x)};
\addlegendentry{\(\frac{1}{\alpha} = \frac{1+ \beta}{\beta}\)};
\addplot [domain=0:2.4, samples=100,densely dotted, color=black ] {1/(1+x)};
\addlegendentry{\(\alpha = \frac{1}{1+\beta}\)};
\addplot [domain=0:2.4, samples=100,densely dashed, color=black] {(2/3)};
\addlegendentry{\(\alpha = \frac{1}{\theta} \)};
\addplot [domain=0:2.4, samples=100,densely dashed, color=black] {1/2};
\addplot [domain=0:2.4, samples=100,densely dashed, color=black] {1/3};
\draw[->, black](axis cs:0.2,2/3) -- (axis cs:0.2,0.626);
\draw[->, black] (axis cs:0.2,0.5) -- (axis cs:0.2, 0.46);
\draw[->, black] (axis cs:0.2,1/3) -- (axis cs:0.2,0.293);
\draw[->, black] (axis cs:0.6,1/1.6) -- (axis cs:0.6, 0.585);
\draw[->, black] (axis cs:1.3, 1/2.3) -- (axis cs:1.3, 0.394);
\end{axis}
\end{tikzpicture}
\hspace*{.5cm}
\begin{tikzpicture}[scale=0.7]
\begin{axis}[ axis lines = left, xlabel = \(\theta\), ylabel style = {rotate=-90, xshift=5pt}, ylabel = {\(\alpha\)}, legend style={yshift=-0.5cm}]
\addplot [domain=1:3, samples=100, densely dashed , color=black] {(5*(x^2)-1)/(4*(x^2))};
\addlegendentry{ \(\alpha = \frac{5 \theta^2 -1}{4 \theta^2}\)};
\addplot [domain=1:3, samples=100, color=black] {3 - (1/x) };
\addlegendentry{\( \alpha = 3 - \frac{1}{\theta}\)};
\end{axis}
\end{tikzpicture}
\caption{{\footnotesize Left: Plot of when the approximation guarantees are equal and parameter constraints for $\alpha$ (dashed and dotted). Right: Lower bound on the approximation guarantee for different values of $\theta$ for the divisible (dashed) and indivisible case.}} \label{fig:plots} 
\end{figure}
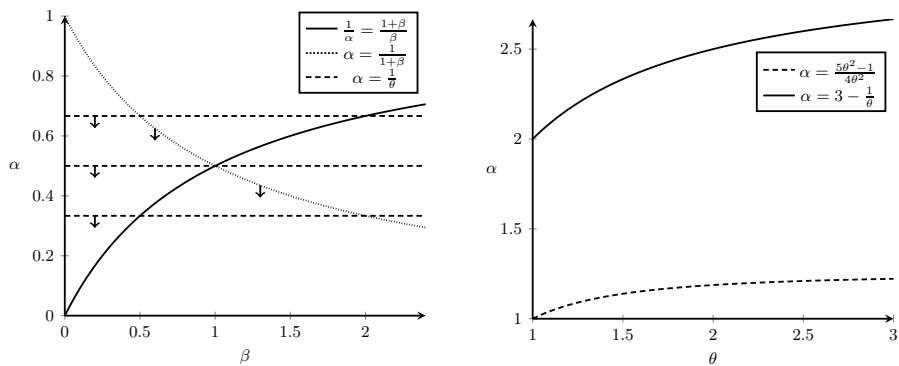

\paragraph{Lower Bound.}
Similar to Section 3, one can show that if $\theta \ge 2$ no deterministic, truthful, individually rational and budget feasible mechanism exists with an approximation guarantee of $(\frac{19}{16} - \epsilon)$, for some $\epsilon>0$. Additionally, this instance can be adjusted to a specific value of $\theta$ to prove that no mechanism exists with an approximation guarantee of $((5 \theta^2 -1)/(4 \theta^2) - \epsilon)$, for some $\epsilon>0$. Details can be found in Appendix \ref{sec:LowBouIndiTheta}. The general lower bound is plotted in Figure \ref{fig:plots} and, as in Section 3, converges to $\frac{5}{4}$ if $\theta \rightarrow \infty$.

For the indivisible case, one can show that if $\theta \ge 2$ no deterministic, truthful, individually rational and budget feasible mechanism exists with an approximation guarantee of $(\frac{5}{2} - \epsilon)$ compared to the fractional optimum, for some $\epsilon>0$. This lower bound can also be adjusted to a specific value of $\theta$ and, as in Gravin et al.~\cite{gravin}, converges to 3 if $\theta \rightarrow \infty$. Again, details can be found in Appendix \ref{sec:LowBouIndiTheta}.

\section{Agent Types with Capped Linear Valuation Functions}

In this section, we extend the original setting with two extra elements: namely agent types and a capped linear valuation function for each type. In the Introduction, we described how this relates to the example in which an auctioneer wants to organize a comedy show. The two extra elements of an instance $\mathcal{I} = (B, (v_i)_{i\in A},(c_i)_{i\in A}, (t_i)_{i\in A}, (l_j)_{j \in T} )$ of this type of procurement auction are defined as follows: Each agent $i \in A$ has a type $t_{i} \in T = \{1, 2, \dots, t \}$, with $t \le n$. If agents have the same type, they are substitutable, meaning that they are offering a similar good or service. For every type $j \in T$, the auctioneer has a valuation function $l_{j}: \mathbb{R}_{\ge 0} \rightarrow \mathbb{R}_{\ge 0}$, which maps the accumulative selected value of type $j$, to the actual value obtained by the auctioneer. We assume that there is a maximum total value of each type that the auctioneer wishes to buy, and the auctioneer will obtain no extra value if more value is selected. Additionally, selecting more than this maximum value will not have a negative impact. More formally for a type $j \in T$, $M_j \in \mathbb{R}^{+}$ represents this maximum value for type $j$, $l_j(x) = x$ for $x\le M_j$ and $l_j(x) = M_j$ for $x \ge M_j$.

In this section, $\opt(A,\mathbf{c})$ now corresponds to the value of the optimal solution of the linear program stated below and $\opt_{-i}(A,\mathbf{c})$ corresponds to the optimal value when regarding the set of agents $A \setminus \{i\}$.

\begin{equation}
\begin{array}[t]{rrcl@{\;\;\;}l}
\max & \sum_{j \in T} l_{j}(V_{j}) \\
\text{s.t.} & \sum_{i \in A} c_i x_i & \le & B & \\
& \sum_{i \in A: t_i = j } v_i x_i & = & V_j & \forall j \in T\\
& V_j & \le & M_j & \forall j \in T\\
& x_i & \in & [0,1] & \forall i \in A
\end{array}
\label{optProgramCapped}
\end{equation}

Note that if an optimal solution has a value of $\sum_{j \in T} M_t$ and some remaining budget, the inequality $V_j \le M_j$ ensures that this remaining budget is not spend on some agent that is not yet entirely selected. This has no influence on optimality, but will be useful when proving budget feasibility and ensures that Lemma \ref{boundOptCost} also holds in this setting. We define for each agent $i$:
\[ \rho_{i}(N,\mathbf{c}) = \frac{ \min \{ v_i, M_{t_i} \} }{\opt_{-i}(N, \mathbf{c})}.\]
Note that an agent $i$ still has no influence on their own ratio. Again we either select one valuable agent, or select agents in a greedy manner according to efficiencies and an optimal solution, which still naturally leads to a monotone allocation rule. Let $\alpha \in (0,1]$ and $\beta > 0$ be some parameters which we fix later. Our mechanism is as follows:

\begin{center}
\begin{boxedminipage}{11cm}
\textsc{\textbf{Divisible Agents Capped Valuation (DA-cap)}}
\begin{algorithmic}[1]
\small
\State Let $N=\{i \in A \; |\; c_i \le B\}$, $n = |N|$ and $i^* = \arg\max_{i\in N} \rho_{i}(N,\mathbf{c}) $
\If{$ \rho_{i^*}(N,\mathbf{c}) \ge \beta$} set $x_{i^*}=1$ and $x_{i}=0$ for $i\in N \setminus \{i^*\}$
\Else
\State Rename agents s.t. $\frac{v_1}{c_1} \ge \frac{v_2}{c_2} \ge \dots \ge \frac{v_n}{c_n}$
\State Compute $\mathbf{x^*}$ of $\opt(N, \mathbf{c})$ and set $\mathbf{x}$ s.t. $v(\mathbf{x}) = \alpha v(\mathbf{x^*}) $ with

$x_{i} = x_{i}^*$ for $i<k$, $x_{k} \in (0,x^*_k]$ and $x_{i} = 0$ for $i>k$, $k \le n$
\For{$i \le k$} 
\If{$c_i > \tau_i $} set $x_i = 0$ \EndIf 
\EndFor
\EndIf
\State For $i \in N$ compute payments $p_i$ according to (\ref{payments})
\State \textbf{return} $(\mathbf{x}, \mathbf{p})$
\end{algorithmic}
\end{boxedminipage}
\end{center}

\begin{theorem} \label{Th:DivAndCap}
Mechanism \textsc{\textbf{DA-cap}} with $\alpha \approx 0.38$ and $\beta \approx 0.61$ is truthful, individually rational, budget feasible, computationally efficient and has an approximation guarantee of $\approx 2.62$.
\end{theorem}

The proof of Theorem \ref{Th:DivAndCap} is similar to earlier proofs and can be found in Appendix \ref{App:proofTheorem}. Note that the lower bound of Section 3 still holds: In both instances used for proving the lower bound one can define that each agent has the same type $t$ and $M_t = \sum_{i \in A} v_i$.

\begin{remark}
Note that the linear program (\ref{optProgramCapped}) can also be formulated with a different objective function, i.e.,  $\sum_{i \in A} v_i x_i $, when combining the second and third constraint, i.e., $\sum_{i \in A: t_i = j } v_i x_i \le M_j, \forall j$. Alternatively, this section can be interpreted as the allocation having to satisfy some combinatorial structure. One property is needed for all the proofs to hold in this case: If the agents are ordered according to decreasing value over cost ratios and an agent moves to the front (back) of the ordering by changing their declared cost, their fraction selected should increase (decrease) or remain the same. This much depends on the imposed combinatorial structure. 
\end{remark}

\section{Agent Types with Concave Valuation Functions}

Besides a maximum total value an auctioneer may want to acquire from a type of agent, the actual value obtained might decrease as increasingly more value is acquired. This is modeled by a function $l_j$ for each type $j \in T$ that is concave, non-decreasing and $l_{j}(0) =0$. 

In this section, $\opt(A,\mathbf{c})$ corresponds to the value of the optimal solution of the concave program stated below, regarding the set of agents $A$, and $\opt_{-i}(A,\mathbf{c})$ regards the set of agents $A \setminus \{i\}$.
\begin{equation*}
\begin{array}[t]{rrcl@{\;\;\;}l}
\max & \sum_{j \in T} l_{j}(V_{j}) \\
\text{s.t.} & \sum_{i \in A} c_i x_i & \le & B & \\
& \sum_{i \in A: t_i = j } v_i x_i & = & V_j & \forall j \in T\\
& x_i & \in & [0,1] & \forall i \in A
\end{array}
\label{optProgramConcave}
\end{equation*}

As the objective function is concave (sum of concave functions), there exists an optimal solution that we can find in polynomial time\footnote[1]{Note that the feasibility region is compact.}. Additionally, there exists an optimal solution $\mathbf{x^*}$ to this program with the following structure. Let $S_t = \{ i \in A : t_i = t\}$ be the set of agents with type $t \in T$. Then when ordering and renaming the agents in this set according to decreasing efficiency, there exists an integer $k \le |S_t|$ such that the fractions selected are $x^*_{i} = 1$ for $i < k$, $x^*_{k} \in [0,1]$ and $x^*_{i} = 0$ for $k < i \le |S_t|$. For agents $i \in \{1, \dots, |S_t| \}$ we define
\[ v^*_i = l_t(x) - l_t(y) \quad \text{with} \quad x = \sum_{j=1}^{i} v_j x^*_j \quad \text{and} \quad y = \sum_{j=1}^{i-1} v_j x^*_j, \]
so $v^*_i = 0$ for agents $i \in \{ k+1, \dots, |S_t| \}$. When we refer to an optimal solution $\mathbf{x^*}$ of $\opt(A,\mathbf{c})$ in this section, we refer to the $\mathbf{x^*}$ with this structure.

Additionally, given a vector of declared costs $\mathbf{c} = (c_i)_{i \in A}$, we define $\hat{v}_i$ for all $i \in A$. Again let $S_t$ be the set of agents with type $t \in T$ and reorder and rename the agents in this set according to decreasing efficiency. For agents $i \in \{1, \dots, |S_t| \}$ we define

\[ \hat{v}_i = l_t(x) - l_t(y) \quad \text{with} \quad x = \sum_{j=1}^{i} v_j \quad \text{and} \quad y = \sum_{j=1}^{i-1} v_j. \]

\paragraph{The Mechanism.}

The ideas behind Mechanism \textsc{\textbf{DA}} can still be applied when adjusted to the current setting. We consider a different ordering of the agents in the second part of the mechanism and alter $\tau_i$ and $\rho_{i}$ to fit the concave valuation functions:
\[ \tau_i = \hat{v}_i \frac{B}{\alpha(1+\beta) \opt_{-i}(N, \mathbf{c})} \quad \text{ and } \quad \rho_{i}(N,\mathbf{c}) = \frac{l_{t_i}(v_i)}{\opt_{-i}(N, \mathbf{c})}.\]
Note that in this mechanism, the threats $\tau_i$ are not independent of the declared cost of agent $i$, but as $\tau_i$ might decrease when agent $i$ declares a higher cost, this imposes no problem. Again, let $\alpha \in (0,1]$ and $\beta > 0$ be some parameters which we fix later. Our mechanism is as follows: 

\begin{center}
\begin{boxedminipage}{11cm}
\textsc{\textbf{Divisible Agents Concave Valuation (DA-con)}}
\begin{algorithmic}[1]
\small
\State Let $N=\{i \in A \;|\; c_i \le B\}$, $n = |N|$ and $i^* = \arg\max_{i\in N} \rho_{i}(N,\mathbf{c})$
\If{$\rho_{i^*}(N,\mathbf{c}) \ge \beta$} set $x_{i^*}=1$ and $x_{i}=0$ for $i\in N \setminus \{i^*\}$
\Else
\State Compute $\mathbf{x^*}$ of $\opt(N, \mathbf{c})$ and corresponding $(v^*_i)_{i \in N}$
\State Rename agents s.t. $\frac{v^*_1}{c_1} \ge \frac{v^*_2}{c_2} \ge \dots \ge \frac{v^*_n}{c_n}$
\State Set $\mathbf{x}$ s.t. $v(\mathbf{x}) = \alpha v(\mathbf{x^*}) $ with $x_{i} = x_{i}^*$ for $i<k$, $x_{k} \in (0,x^*_k]$ 

and $x_{i} = 0$ for $i>k$, $k \le n$
\For{$i \le k$} 
\If{$c_i > \tau_i $} set $x_i = 0$ \EndIf 
\EndFor
\EndIf
\State For $i \in N$ compute payments $p_i$ according to (\ref{payments})
\State \textbf{return} $(\mathbf{x}, \mathbf{p})$
\end{algorithmic}
\end{boxedminipage}
\end{center}

\begin{lemma}\label{TrIndRat3}
Mechanism \textsc{\textbf{DA-con}} is truthful and individually rational. 
\end{lemma}

\begin{proof}
Proving that the mechanism is truthful and individually rational is identical to the proof of Lemma \ref{TrInRa}, when using the new definition of $\rho_{i}(N,\mathbf{c})$. Note that by the way we define $\mathbf{x^*}$, if in the second part of the mechanism an agent increases (decreases) their declared cost, their fraction selected can only decrease (increase) or remain the same. Additionally, note that the threat will not deselect an agent $i$ if $i$ decreases their cost, as then the threat can only increase or remain the same. \hfill \qed 
\end{proof}

\begin{lemma}\label{BudFea3}
Mechanism \textsc{\textbf{DA-con}} is budget feasible.
\end{lemma}

\begin{proof}
If, given declared costs $\mathbf{c} = (c_i)_{i \in A}$, $\mathbf{x}$ is computed by selecting agent $i^*$, the proof is identical to Lemma \ref{BudFea}. Otherwise the allocation vector is computed in the second part of the mechanism. By construction we know that the threshold bid for agent $i$ is smaller than or equal to our threat $\tau_i$, as $\tau_i$ can only decrease or remain the same when agent $i$ increases their declared cost. Therefore, and because the allocation rule is monotone, the payment of agent $i$ can be bounded by $x_i \tau_i$ and
\[\sum_{i\in N} p_i = \sum_{i=1}^{k} p_{i} \le \sum_{i=1}^{k} x_i \hat{v}_i \frac{B}{\alpha(1+\beta) \opt_{-i}(N, \mathbf{c})} \le \sum_{i=1}^{k} x_i \hat{v}_i \frac{B}{\alpha \opt(N, \mathbf{c})} \le B,\]
were the second inequality follows from $\opt(N, \mathbf{c}) \le \opt_{-i}(N, \mathbf{c}) + l_{t_i}(v_i)  < (1 + \beta)\opt_{-i}(N, \mathbf{c})$, as the mechanism did not select agent $i^*$. The last inequality follows as the functions $l_t$ are concave and $x_i \le x_i^*$ by construction, so $x_i \hat{v}_i \le v^*_i$. Hence, the mechanism is budget feasible. \hfill \qed 
\end{proof}

In this setting Lemma \ref{LemLowBouB} does not hold and a different version can be used to prove that the treats will not discard any agents. As in Section 3, the Lemma will only be used if \textsc{\textbf{DA-con}} computed the allocation vector in a greedy manner and therefore it is additionally assumed that the $\rho_{i}$'s are bounded by $\beta$.  

\begin{lemma}\label{LemLowBouB2}
Let $\mathbf{x^*}$ be the corresponding solution of $\opt(A, \mathbf{c})$ in which at most one agent of each type is fractionally selected (note that such a solution always exists). Assume agents $i\in A$, $n = |A|$, are ordered such that $\frac{v^*_1}{c_1} \ge \dots \ge \frac{v^*_n}{c_n}$ and $\rho_{i^*}(A, \mathbf{c}) \le \beta$. Let $k\le n$ be an integer such that $\sum_{i=1}^{k-1} v^*_i < \alpha \opt(A, \mathbf{c}) \le \sum_{i=1}^{k} v^*_i$, $\alpha \in (0,1]$. Then $\frac{c_k}{v^*_k} (1 - \alpha - t \beta) \opt(A, \mathbf{c}) \le B$.
\end{lemma}

\begin{proof}
For convenience let $\opt = \opt(A, \mathbf{c})$. Note that we can split $\mathbf{x}^*$ into $\mathbf{x} = (x^{*}_{1}, \dots, x^{*}_{k-1}, 0, \dots, 0)$ and $\mathbf{y} = (0, \dots, 0, x^{*}_{k}, x^{*}_{k+1}, \dots, x^{*}_{n})$ such that $\mathbf{x}^* = \mathbf{x} + \mathbf{y}$, $\sum_{i=1}^{k-1} v^*_i < \alpha \opt$ and $\sum_{i=k}^{n} v^*_i > (1 - \alpha) \opt$. Let $S = \{ k \le i \le n: 0 < y_i < 1 \}$. By feasibility of $\mathbf{x}^*$, and thus $\mathbf{y}$, and the ordering of the agents it follows that:
\begin{align*}
B \geq c(\mathbf{y}) 
& = \sum_{i=k}^{n} \frac{v_i^*}{v_i^*} c_i y_i \geq \frac{c_k}{v_k^*} \sum_{i=k}^{n} v_i^*y_i 
=  \frac{c_k}{v_k^*} \bigg( \sum_{i=k}^{n} v_i^* - \sum_{i \in S} (1- y_i) v_i^* \bigg) \\
& > \frac{c_k}{v_k^*} \bigg( (1-\alpha)\opt - \sum_{i \in S} v_i^* \bigg) 
\ge \frac{c_k}{v_k^*} ( (1-\alpha)\opt - \sum_{i \in S} \beta \opt_{-i} ) \\
& \ge \frac{c_k}{v_k^*} (1-\alpha - t \beta)\opt, 
\end{align*}
as $v^*_i = 0$ for agents with $x^*_{i} = 0$ (second equality) and $\rho_{i^{*}}(A, \mathbf{c}) \le \beta$, so $\forall i$: $v_i^* \le l_{t_i}(v_i) \le \beta \opt_{-i}$. The last inequality follows as at most one agent of each type is fractionally selected in $\mathbf{x^*}$, and thus in $\mathbf{y}$. \hfill \qed 
\end{proof}

\begin{lemma}\label{AppGua3}
Mechanism \textsc{\textbf{DA-con}} has an approximation guarantee of $\frac{1 + \beta}{\beta}$  if $\alpha = \frac{\beta}{1+\beta}$ and 
\begin{equation}
\beta = \frac{\sqrt{t^2 + 6t +5}}{2(1+t)} - \frac{1}{2}.
\label{beta}
\end{equation}
\end{lemma}

\begin{proof}
Given declared costs $\mathbf{c} = (c_i)_{i \in A}$, let $\mathbf{x}$ be computed by the mechanism. If the mechanism selected agent $i^*$, we have
\[ l_{t_i^*}(v_{i^*}) \ge \beta \opt_{-i^*}(N, \mathbf{c}) \ge \beta \opt(N, \mathbf{c}) - \beta l_{t_i^*}(v_{i^*}) \Leftrightarrow  v(\mathbf{x}) \ge \frac{\beta}{1 + \beta} \opt(N, \mathbf{c}).\]

Otherwise the mechanism computed $\mathbf{x}$ in a greedy manner. We want no agent $i \le k$ to be deselected by our threat, i.e., we want $c_i \le \tau_i$. If  
\[ R:= \frac{\alpha(1+\beta)}{(1 - \alpha - t \beta)} \le 1 \quad \text{ then } \quad \frac{\alpha(1+\beta)}{(1 - \alpha - t \beta)} \frac{v^*_i}{\hat{v}_i}  \opt_{-i}(N, \mathbf{c}) \le \opt(N, \mathbf{c}),\]
as $\frac{v^*_i}{\hat{v}_i} \le 1$ by construction. Rearranging terms leads to
\[ \tau_i = \hat{v}_i \frac{B}{\alpha(1+\beta) \opt_{-i}(N, \mathbf{c})} \ge  v^*_i \frac{B}{(1 - \alpha - t \beta ) \opt(N, \mathbf{c})} \ge c_i, \]
where the last inequality follows from Lemma~\ref{LemLowBouB2} as \textsc{\textbf{DA-con}} computed $\mathbf{x^*}$ such that most one agent of each type is fractionally selected. Balancing the approximation guarantees $\frac{1 + \beta}{\beta} = \frac{1}{\alpha}$ subject to $R \le 1$ leads to the desired result. \qed  
\end{proof}

\begin{theorem}\label{TH:DA-CV}
Mechanism \textsc{\textbf{DA-con}} with $\beta$ as in (\ref{beta}) and $\alpha = \frac{\beta}{1 + \beta}{}$ is truthful, individually rational, budget feasible and has an approximation guarantee of $\frac{1 + \beta}{\beta}$.
\end{theorem}

Theorem \ref{TH:DA-CV} follows from Lemmas \ref{TrIndRat3}, \ref{BudFea3} and \ref{AppGua3}. The combination of concave valuation functions and the possibility of multiple fractional agents in $\mathbf{x}^*$ made it increasingly difficult to find a good threat. As a result, the approximation guarantee grows linearly with the number of agent types, and some approximation guarantees for small values of $t$ can be found in Table \ref{Tab:concaveAppGua}.

\begin{table}[t]
\setlength{\tabcolsep}{7pt}
\renewcommand{\arraystretch}{1.2}
\begin{center}
\caption{{\footnotesize Approximation guarantee for $t = 2, \dots, 6$.}}
\begin{tabular}{|c||c|c|c|c|c|} 
\hline
$t$ & 2 & 3 & 4 & 5 & 6 \\
\hline 
 $\frac{1 + \beta}{\beta}$ & 4.80 & 5.83 & 6.86 & 7.88 & 8.89\\
\hline
\end{tabular}
\label{Tab:concaveAppGua}
\end{center}
\vspace*{-.5cm}
\end{table}

When an instance has multiple agent types, there are (rather artificial) cases where the functions $l_j$ may intersect infinitely many times and the number of breakpoints of the payment function is not polynomially bounded. Mechanism \textsc{\textbf{DA-con}} remains computationally efficient if the number of such breakpoints is polynomially bounded. Note that the lower bound of Section 3 still holds: In both instances used for proving the lower bound one can define that each agent has the same type $t$ and $l_{t}(x) = x$.

\paragraph{One agent type:} Note that if there is only one type of agent, we can still use mechanism \textsc{\textbf{DA}}, where $\opt$ and $\opt_{-i}$ are still w.r.t. an additive linear valuation function. Let $\mathbf{x}$ and $\mathbf{x^*}$ be the allocations of mechanism \textsc{\textbf{DA}} and the optimal fractional solution respectively. By Theorem \ref{Th:DA}, we have $\frac{v(\mathbf{x})}{v(\mathbf{x^*})} \ge \frac{1}{2.62}$. Let $l(\cdot)$ be the concave valuation function, which, for convenience, now maps an allocation vector to a real number. Note that $\mathbf{x^*}$ is also an optimal fractional solution in the case of one concave valuation function. Let $l(\mathbf{x}) = k_1 \cdot v(\mathbf{x})$ and $l(\mathbf{x^*}) = k_2 \cdot v(\mathbf{x^*})$. We have 
\[ \frac{l(\mathbf{x})}{l(\mathbf{x^*})} = \frac{k_1 \cdot v(\mathbf{x}) }{k_2 \cdot v(\mathbf{x^*})} \ge  \frac{v(\mathbf{x})}{v(\mathbf{x^*})} \ge \frac{1}{2.62},\]
as $\frac{k_1}{k_2} \ge 1$ by concavity of $l(\cdot)$.

\section{Conclusion and Future Work}

We considered budget feasible mechanisms for procurement auctions in the case of divisible agents. The challenge of exploiting the divisibility of agents is to limit the consequences of selecting agents fractionally, so that budget feasibility can be ensured. To achieve this in the initial setting with an additive linear valuation function, we introduced the notion of threats and found an upper bound of 2.62 and a lower bound of 1.25 on the approximation guarantee. When introducing a notion of competitiveness between agents, we found an upper and lower bound on the approximation guarantee of 2 and 1.18 respectively. In both settings the gap between the upper and lower bound on the approximation guarantee remains and it would be interesting to see if both of these gaps can be tightened or even closed. Both results can be extended to the setting with multiple agent types and a linear capped valuation function for each type. This setting can be used to model a maximum total value of each type that the auctioneer wishes to buy, but selecting more value will never have a negative impact.

An interesting extension of our multiple agent type setting is to consider concave (non-decreasing) valuation functions. Such functions can be used to model diminishing returns of the auctioneer, i.e., the marginal increase in value decreases as more value of a certain type is selected. Our preliminary result shows that the general mechanism can be adapted to this more general setting, though at the expense of a worse approximation guarantee (grows linearly with the number of agent types). As in the setting with linear capped valuation functions, multiple agents can be selected fractionally, but the concave valuation functions increase the difficulty of finding a good threat. 

It would be natural to also study the divisible case in settings with other valuation functions. Additionally, one could study the problem with a required structure on the allocation to capture relations between agents.

\subsubsection{Acknowledgements.} We thank Georgios Amanatidis for proposing to study budget feasible mechanisms for divisible agents when he was a postdoc at CWI. Part of this work was sponsored by the Open Technology Program of the
Dutch Research Council (NWO), project number 18938.

\appendix

\section{Computationally efficient}\label{AppComEff}

We explain why the mechanism \textsc{\textbf{DA}} is computationally efficient. This is trivial for the majority of the mechanism, but may not be for the computation of the payment vector. Computing the payments according to (\ref{payments}) can also be done in polynomial time, as these payment functions are piecewise and have a polynomially bounded number of breakpoints. Note that $p_i = 0$ if $x_i=0$. 

If the mechanism selects agent $i^*$, the payment function of $i^*$ only has one breakpoint, namely at $B$ or at some $c' < B$ for which another agent becomes $i^*$. Otherwise agents are selected such that $v(\mathbf{x}) = \alpha \opt$. We will consider the maximum number of breakpoints the payment function can have in this case for some agent $i$. Suppose the agents are ordered according to decreasing efficiency, and let $j \le n$ be the highest index of an agent selected in $\opt$. If initially $x_i =1$, agent $i$ can increase their declared cost until $i$ becomes agent $k$ and is selected fractionally. This is the first breakpoint of the payment function. If $i$ then continues to increase their declared cost, $i$ will at some point move a place up in the ordering and could then still be fractionally selected. This is another breakpoint of the payment function, as $i$'s fraction selected decreases by a substantial amount at this point. When $i$ increases their declared cost, the index $j$ also decreases at some point. This can be another breakpoint of the payment function, as the slope might then change. When $i$ continues to increase their cost, at some point $i$'s place in the ordering will be equal to the index $j$ of that moment. In this case $i$ will not be selected as the fractional agent as then we would have $v_i > (1-\alpha) \opt \ge \beta \opt_{-i}$, by the values of $\alpha$ and $\beta$. The number of breakpoints of the payment function of agent $i$ is therefore bounded by $j-i \le n-1$. Note that the number of breakpoints can be smaller than $j-i$ by the threat we impose or by the mechanism selecting $i^*$. Additionally, the sub-functions of the payment function are well-defined and can be found efficiently. 

Note that for the mechanism \textsc{\textbf{DA-$\theta$}}, the number of breakpoints of the payment function of an agent $i$ is also bounded by $n-1$. If agent $i=1$ would move all the way to the back of the ordering, $i$ will not be selected as then $v_i > (1-\alpha) \opt \ge \beta \opt_{-i}$, by the values of $\alpha$ and $\beta$.

\section{Lower Bounds Section 4 }\label{sec:LowBouIndiTheta}

\paragraph{Divisible Case:} Next we show that if $\theta \ge 2$ no deterministic, truthful, individually rational and budget feasible mechanism exists with an approximation guarantee of $(\frac{19}{16} - \epsilon)$, for some $\epsilon>0$. For contradiction, assume such a mechanism does exist. Consider the following two instances, both with budget $B$ and two agents with equal valuations: \[\mathcal{I}_1 = \Big (B, \mathbf{v} = (1, 1), \mathbf{c} = \Big (\frac{2}{3}B,\frac{2}{3}B \Big ) \Big)  \text{ and } \mathcal{I}_2 = \Big (B, \mathbf{v} = (1, 1), \mathbf{c} = \Big (\frac{1}{3}B,\frac{2}{3}B \Big ) \Big ).\]
Note that both instances satisfy (\ref{eqTheta}) for $\theta \ge 2$. In the first instance $\opt = 1.5$, so for the approximation guarantee to hold, an allocation vector must satisfy $v(\mathbf{x})  > \frac{24}{19}$. Therefore, any such mechanism must have some agent $i$ for which $x_i > \frac{12}{19}$ and assume w.l.o.g. that this is agent 1. In the second instance $\opt = 2$, so for the approximation guarantee to hold, an allocation vector must now satisfy $v(\mathbf{x}) > \frac{32}{19}$. By the previous instance and individual rationality, it follows that agent 1 can guarantee itself a utility of at least $\frac{4}{19}B$ by deviating to $\frac{2}{3}B$. As agent 1 must be somewhat selected to achieve the approximation guarantee, we have $p_1 > \frac{1}{3}B x_1+ \frac{4}{19} B$ by truthfulness. Therefore in the best case, if agent 1 is entirely selected, this leads to a budget left smaller than $\frac{26}{57} B$. By spending this on agents 2, the value that can be acquired is smaller than $\frac{39}{57}$, leading to an allocation vector with value smaller than $\frac{96}{57}$ and resulting in a contradiction.

Note that this instance can be adjusted to a specific value of $\theta$ as follows:
\[\mathcal{I}_1 = \Big (B, \mathbf{v}, \mathbf{c} = \Big(\frac{\theta}{\theta +1}B,\frac{\theta}{\theta +1}B \Big ) \Big )  \text{ and }  \mathcal{I}_2 = \Big (B, \mathbf{v}, \mathbf{c} = \Big(\frac{1}{\theta +1}B,\frac{\theta}{\theta +1}B \Big ) \Big),\]
to prove that no mechanism exists with an approximation guarantee of $((5 \theta^2 -1)/(4 \theta^2) - \epsilon)$, for some $\epsilon>0$.

\paragraph{Indivisible Case:} Using two almost similar instances as in Gravin et al.~\cite{gravin}, we show that if $\theta \ge 2$ in the indivisible case no deterministic, truthful, individually rational and budget feasible mechanism exists with an approximation guarantee of $(\frac{5}{2} - \epsilon)$ compared to the fractional knapsack optimum, for some $\epsilon>0$. For contradiction, assume such a mechanism does exist. Consider the following two instances, both with budget B and 3 agents with equal valuations:
\[\mathcal{I}_1 = (B, \mathbf{v} = \mathbf{1}, \mathbf{c} = (c^*,c^*, c^*)) \quad \text{and} \quad \mathcal{I}_2 = (B, \mathbf{v} = \mathbf{1}, \mathbf{c} = (c^* /2,c^*, c^*)),\]
with $c^* = B/(2 - \frac{\epsilon}{2})$. Note that both instances satisfy (\ref{eqTheta}) for $\theta \ge 2$. In the first instance $\opt = 2 - \frac{\epsilon}{2}$. A mechanism can only satisfy all the properties if one agent is (fully) selected and this agent is paid at least $c^*$. Assume w.l.o.g. that this is agent 1. In the second instance $\opt = 2.5 - \frac{\epsilon}{2}$, so two agents must be (fully) selected to achieve the approximation guarantee. By truthfulness, the payment of agent 1 must be at least $c^*$, as otherwise agent 1 could have been better off declaring a cost of $c^*$. As the payment to the other agent (2 or 3) that is (fully) selected must also be at least $c^*$ by individual rationality, this results in a contradiction. 

Note that this instance can be adjusted to a specific value of $\theta$ by setting $c_1 = c^* / \theta$ in the second instance to prove that no deterministic, truthful, individually rational and budget feasible mechanism exists with an approximation guarantee of $(3 - \frac{1}{\theta} - \epsilon)$, for some $\epsilon>0$.

\section{Proof of Theorem \ref{Th:DivAndCap}
}\label{App:proofTheorem}

\begin{proof}
Proving that the mechanism is truthful and individually rational is identical to the proof of Lemma \ref{TrInRa}, when the definition $\rho_{i}(N,\mathbf{c}) = \frac{\min \{ v_i, V_{t_i}  \}}{\opt_{-i}(N, \mathbf{c})}$, $i \in N$, is used. Note that by the way we define $\mathbf{x^*}$, if in the second part of the mechanism an agent moves to the front (back) of the ordering, their fraction selected can only increase (decrease) or remain the same.

Proving that the mechanism is budget feasible is identical to the proof of Lemma \ref{BudFea}, except for the reasoning of the last two inequalities. The last inequality follows from the construction of $x^*$ and the second to last inequality follows as 
\[\opt(N, \mathbf{c}) \le \opt_{-i}(N, \mathbf{c}) + \min \big \{v_i,V_{t_i} \big \}  < (1 + \beta)\opt_{-i}(N, \mathbf{c}).\]

As Lemma \ref{LemLowBouB} still holds, the approximation guarantee follows from the proof of Lemma \ref{AppGua}, were the reasoning behind the approximation guarantee when agent $i^*$ is selected is slightly different:
\[ \min \{ v_{i^*}, V_{t_{i^*}}\} \ge \beta \opt_{-i^*}(N, \mathbf{c}) \ge \beta \opt(N, \mathbf{c}) - \beta \min \big \{ v_{i^*}, V_{t_{i^*}} \big \}  \quad \Leftrightarrow \] 
\[ v(\mathbf{x}) =  \min \big \{ v_{i^*}, V_{t_{i^*}} \big \} \ge \frac{\beta}{1 + \beta} \opt(N, \mathbf{c}).\]

The same reasoning of computational efficiency as in Appendix \ref{AppComEff} is applicable, which, altogether, proofs the theorem.  \hfill \qed
\end{proof}

\end{document}